\documentclass[letterpaper, 12pt]{article}
\usepackage[utf8]{inputenc}
\usepackage[paperwidth=8.5in, paperheight=11in, margin=1in]{geometry}

\pdfoutput=1

\usepackage[activate={true,nocompatibility},final,tracking=true,kerning=true,spacing=true,factor=1100,stretch=10,shrink=10]{microtype}
\microtypecontext{spacing=nonfrench}

\usepackage{tikz}
\usepackage{tkz-graph}
\tikzstyle{vertex}=[circle, draw, inner sep=0pt, minimum size=6pt]

\usepackage{verbatim}
\usepackage{hyperref}
\usepackage{caption}
\usepackage{subcaption}

\usepackage{graphicx}
\graphicspath{ {images/} }
\usepackage[fleqn]{amsmath}

\usepackage{amsthm}
\usepackage{amssymb}
\usepackage{enumitem}
\usepackage{textcomp}
\usepackage{atbegshi,picture}
\usepackage{amssymb}

\usepackage{array}
\newcolumntype{$}{>{\global\let\currentrowstyle\relax}}
\newcolumntype{^}{>{\currentrowstyle}}

\usepackage{relsize}

\usepackage{breqn}

\newcommand{\myN}{\mathcal{F}^{old}}
\newcommand{\myL}{\mathcal{F}^{new}_1}
\newcommand{\myM}{\mathcal{F}^{new}_0}
\newcommand{\mySa}{S^{new}_1}
\newcommand{\mySb}{S^{new}_0}

\makeatletter
\def\thm@space@setup{%
  \thm@preskip=0.5cm
}
\makeatother

\newtheorem{thm}{Theorem}

\newtheorem{lem}[thm]{Lemma}

\newtheorem{inclm}{Claim}[thm]

\makeatletter
\newenvironment{inproof}[1][\proofname]{\par
  \pushQED{\qed}%
  \normalfont \partopsep=\z@skip \topsep=\z@skip
  \trivlist
  \item[\hskip\labelsep
        \itshape
    #1\@addpunct{.}]\ignorespaces
}{%
  \popQED\endtrivlist\@endpefalse
}
\makeatother

\title{A $\frac{17}{12}$-approximation algorithm for 2-vertex-connected spanning subgraphs on graphs with minimum degree at least 3}
\author{Vishnu V. Narayan}
\date{January 17 2017\footnote{The research was completed on July 31, 2016; this draft was delayed due to other commitments.}}

\begin{document}

\maketitle

We obtain a polynomial-time $\frac{17}{12}$-approximation algorithm for the minimum-cost 2-vertex-connected spanning subgraph problem, restricted to graphs of minimum degree at least 3. Our algorithm uses the framework of ear-decompositions for approximating connectivity problems, which was previously used in algorithms for finding the smallest 2-edge-connected spanning subgraph by Cheriyan, Seb\H{o} and Szigeti (SIAM J.Discrete Math. 2001) who gave a $\frac{17}{12}$-approximation algorithm for this problem, and by Seb\H{o} and Vygen (Combinatorica 2014), who improved the approximation ratio to $\frac{4}{3}$.

\section*{Introduction}

A graph is 2-vertex-connected if the deletion of any vertex, along with its incident edges, does not disconnect the remainder of the graph. The problem of finding a smallest 2-vertex-connected spanning subgraph of a given graph is NP-hard. This can be seen via the following reduction from the Hamiltonian cycle problem: A graph $G$ has a Hamiltonian cycle if and only if the number of edges in the smallest 2-vertex-connected spanning subgraph is equal to the number of vertices in the graph.

Khuller and Vishkin gave a $\frac{5}{3}$-approximation algorithm for 2-vertex-connectivity in \cite{khullervishkin}. This was improved by Garg, Singla and Vempala, who obtained an approximation ratio of $\frac{3}{2}$ in \cite{gargsinglavempala}. Better approximation ratios have been claimed in the past, but to the best of our knowledge, no complete proof had been published for these. Recently, Heeger and Vygen gave a $\frac{10}{7}$-approximation algorithm for this problem. Our research was carried out independently in the same period.

We present a $\frac{17}{12}$-approximation algorithm for the 2-vertex-connectivity problem restricted to graphs with minimum degree at least 3. Appendix A contains a proof that this restricted version of the problem is also NP-hard. Our algorithm uses the framework of ear-decompositions for approximating connectivity problems, which was previously used in algorithms for finding the smallest 2-edge-connected spanning subgraphs by Cheriyan, Seb\H{o} and Szigeti in \cite{cheriyanseboszigeti} who gave a $\frac{17}{12}$-approximation algorithm for this problem, and by Seb\H{o} and Vygen in \cite{sebovygen}, who improved the approximation ratio to $\frac{4}{3}$.

\section*{Preliminaries}

Let $G = (V,E)$ be an undirected graph. An \textit{ear} of $G$ is a path $P$ of length at least 1, such that the endpoints of $P$ may coincide, but every other pair of vertices of $P$ are distinct. An ear $P$ is \textit{open} if its endpoints are distinct and \textit{closed} otherwise. $P$ is \textit{trivial} if it has a single edge, \textit{short} if it has 2 or 3 edges, and \textit{long} otherwise. $P$ is \textit{even} if it has an even number of edges, and \textit{odd} otherwise. The vertices of $P$ that are not endpoints of $P$ are called internal vertices of $P$, their set is denoted by $in(P)$.

An \textit{ear-decomposition} of $G$ is a sequence $P_0$,$P_1$,\ldots,$P_k$, where $P_0$ is a vertex and $P_1$,\ldots,$P_k$ are ears such that $P_i$ shares exactly its two endpoints with the vertices of $P_0 \cup\ldots\cup P_{i-1}$. We denote by $\phi(G)$ the minimum number of even ears of an ear-decomposition $D$, over all the ear-decompositions $D$ of $G$. An ear-decomposition is \textit{evenmin} if the number of even ears is equal to $\phi(G)$. For any ear $P$, let $\phi(P) = 1$ if $P$ is even and $\phi(P) = 0$ otherwise. A nontrivial ear $P$ is a \textit{pendant} ear if no other nontrivial ear has an endpoint in $in(P)$, otherwise it is non-pendant.

We refer the reader to \cite{sebovygen} for definitions and a detailed discussion of \textit{nice} ear-decompositions, \textit{eardrums}, \textit{earmuffs} and \textit{maximum earmuffs}. We also use the lower bounds $LP(G)$ and $L_\mu(G,M)$ and their related theorems, as defined in Section 4 of this paper. We denote by $OPT_{2VC}(G)$ the cost of the minimum-cost 2-vertex-connected spanning subgraph of $G$.

At times, we abuse the notation for trivial ears, and write $uv$ for the ear corresponding to the path containing the vertices $u$, $v$ and the edge $uv$.

\section*{Algorithm}

Our algorithm consists of a few steps, summarized as follows.
\begin{enumerate}
    \item Construct an open evenmin ear-decomposition $D$ of $G$.
    \item Modify $D$ to get an open evenmin ear-decomposition with the property that all of its short ears are pendant ears.
    \item Modify $D$ to get an open evenmin ear-decomposition that is nice.
    \item Delete all edges in trivial ears. The resulting graph is a 2-vertex-connected spanning subgraph of $G$ with at most $\frac{17}{12}OPT_{2VC}(G)$ edges.
\end{enumerate}

Our analysis is detailed in Theorems \ref{pendanttheorem}, \ref{opennicetheorem} and \ref{algorithmtheorem} below.

The following Lemma (Lemma \ref{firstearlemma}) allows us to replace a given ear-decomposition on a graph with another ear-decomposition on the same graph. It is used in the proof of Theorem \ref{pendanttheorem}. The Lemma is simple to prove, and its proof is left to the reader. The reader may find Figure \ref{lemma1figure} useful for understanding the statement of the Lemma.

\begin{lem} \label{firstearlemma}
    Let $D$ be an ear-decomposition of a graph $G$. Suppose $P$ and $Q$ are nontrivial ears of $D$ such that $Q$ is the first nontrivial ear of $D$ with an endpoint in an internal vertex of $P$. Further, suppose that only one of the endpoints of $Q$ is an internal vertex of $P$, and that this endpoint is adjacent, by an edge of $P$, to an endpoint of $P$. Let $x$ and $y$ be the end vertices of $P$ and $w$ and $z$ be the end vertices of $Q$, such that $w$ is the internal vertex of $P$ adjacent to $y$.
    
    Let $P'$ be the ear with endpoints $x$ and $z$ and consisting of all edges of $P$ and $Q$ except $wy$. Let $D'$ be the ear-decomposition constructed from $D$ by deleting the ears $P$ and $Q$, adding the ear $P'$ in the position of $Q$, and adding the trivial ear $wy$ at the end of the ear-decomposition. Then $D'$ is a valid ear-decomposition of $G$.
\end{lem}

\begin{figure}[ht]
    \begin{subfigure}{0.47\textwidth}
        \centering
        \begin{tikzpicture}[scale=0.5]
        \begin{scope}[every node/.style={circle, fill=black, draw, inner sep=0pt,
        minimum size = 0.2cm
        }]
            \node[label={[label distance=5]225:x}] (x) at (0,0) {};
            \node (v) at (0.5,2) {};
            \node[label={[label distance=5]105:w}] (w) at (2,2) {};
            \node[label={[label distance=5]315:y}] (y) at (2.5,0) {};
            \node (n1) at (3,3) {};
            \node (n2) at (5,3) {};
            \node[label={[label distance=5]45:z}] (z) at (6,2) {};
            
        \end{scope}
        \begin{scope}[every edge/.style={draw=black}]
        \path[very thick] (x) edge node {} (v);
        \path[very thick, dotted] (v) edge[bend left=40] node {} (w);
        \path[very thick] (w) edge node {} (y);
        \path[very thick] (w) edge node {} (n1);
        \path[very thick, dotted] (n1) edge[bend left=40] node {} (n2);
        \path[very thick] (n2) edge node {} (z);
        \end{scope}
        \begin{scope}[every node/.style={draw=none,rectangle}]
        \node (p) at (1.25,0) {$P$};
        \node (q) at (4,2) {$Q$};
        \end{scope}
        \end{tikzpicture}
    \caption{}
    \label{firstlemmaa}
    \end{subfigure}
    \hspace*{\fill}
    \begin{subfigure}{0.47\textwidth}
        \centering
        \begin{tikzpicture}[scale=0.5]
        \begin{scope}[every node/.style={circle, fill=black, draw, inner sep=0pt,
        minimum size = 0.2cm
        }]
            \node[label={[label distance=5]225:x}] (x) at (0,0) {};
            \node (v) at (0.5,2) {};
            \node[label={[label distance=5]105:w}] (w) at (2,2) {};
            \node[label={[label distance=5]315:y}] (y) at (2.5,0) {};
            \node (n1) at (3,3) {};
            \node (n2) at (5,3) {};
            \node[label={[label distance=5]45:z}] (z) at (6,2) {};
            
        \end{scope}
        \begin{scope}[every edge/.style={draw=black}]
        \path[very thick, dashed] (x) edge node {} (v);
        \path[very thick, dashed] (v) edge[bend left=40] node {} (w);
        \path[very thick, dotted] (w) edge node {} (y);
        \path[very thick, dashed] (w) edge node {} (n1);
        \path[very thick, dashed] (n1) edge[bend left=40] node {} (n2);
        \path[very thick, dashed] (n2) edge node {} (z);
        \end{scope}
        \begin{scope}[every node/.style={draw=none,rectangle}]
        \node (p') at (4,2) {$P'$};
        \end{scope}
        \end{tikzpicture}
    \caption{}
    \label{firstlemmab}
    \end{subfigure}
    \caption{}
    \label{lemma1figure}
\end{figure}

\begin{thm} \label{pendanttheorem}
    Every 2-vertex-connected graph $G$ with minimum degree at least 3 has an open ear-decomposition with $\phi(G)$ even ears in which all short ears are pendant. Such an ear-decomposition can be computed in polynomial time.
\end{thm}
\begin{proof}
    Using Proposition 3.2 of Cheriyan, Seb\H{o} and Szigeti \cite{cheriyanseboszigeti}, construct an open ear-decomposition $D = (P_0, P_1, \ldots, P_k)$ of $G$ with $\phi(G)$ even ears.
    
    Suppose the closed ear $P_1$ is short (that is, $P_1$ is a 3-ear). Since every vertex of $G$ has degree at least 3, $G$ has at least 4 vertices, hence $D$ has at least one open ear. Suppose $u$ and $v$ are the end vertices of $P_2$, then there is a $u,v$-path of length 2 in $P_1$. Let $P'$ be the union of the $u,v$-path in $P_1$ and the $u,v$-path in $P_2$. Delete the ears $P_1$ and $P_2$ from $D$, and add the ear $P'$ in the position of $P_1$ in $D$, and the trivial ear $uv$ at the end of $D$. Now the closed ear in $D$ is a long ear, and $D$ is still evenmin. Set $k := k-1$ and relabel the new ears of $D$ such that $D = (P_0, P_1, \ldots, P_k)$.
    
    \begin{figure}[ht]
        \begin{subfigure}{0.47\textwidth}
            \centering
            \begin{tikzpicture}[scale=0.5]
            \begin{scope}[every node/.style={circle, fill=black, draw, inner sep=0pt,
            minimum size = 0.2cm
            }]
                \node[fill=none, thick] (x) at (1,0) {};
                \node[label={[label distance=5]165:u}] (u) at (0,2) {};
                \node[label={[label distance=5]15:v}] (v) at (2,2) {};
            \end{scope}
            \begin{scope}[every edge/.style={draw=black}]
            \path[very thick] (x) edge node {} (u);
            \path[very thick] (u) edge node {} (v);
            \path[very thick] (v) edge node {} (x);
            \path[very thick, dotted] (u) edge[min distance=80, bend left=120] node[label={[label distance=0]90:$P_2$}] {} (v);
            \end{scope}
            \begin{scope}[every node/.style={draw=none,rectangle}]
            \node (pzero) at (2,0) {$P_0$};
            \node (pone) at (2.5,1) {$P_1$};
            \end{scope}
            \end{tikzpicture}
        \caption{}
        \label{cycleeara}
        \end{subfigure}
        \hspace*{\fill}
        \begin{subfigure}{0.47\textwidth}
            \centering
            \begin{tikzpicture}[scale=0.5]
            \begin{scope}[every node/.style={circle, fill=black, draw, inner sep=0pt,
            minimum size = 0.2cm
            }]
                \node[fill=none, thick] (x) at (1,0) {};
                \node[label={[label distance=5]165:u}] (u) at (0,2) {};
                \node[label={[label distance=5]15:v}] (v) at (2,2) {};
            \end{scope}
            \begin{scope}[every edge/.style={draw=black}]
            \path[very thick, dashed] (x) edge node {} (u);
            \path[dashed] (u) edge node {} (v);
            \path[very thick, dashed] (v) edge node {} (x);
            \path[very thick, dashed] (u) edge[min distance=80, bend left=120] node[] {} (v);
            \end{scope}
            \begin{scope}[every node/.style={draw=none,rectangle}]
            \node (pzero) at (2,0) {$P_0$};
            \node (pone) at (3.5,2) {$P_1$};
            \end{scope}
            \end{tikzpicture}
        \caption{}
        \label{cycleearb}
        \end{subfigure}
        \caption{The cycle-ear $P_1$ is a short ear}
    \end{figure}
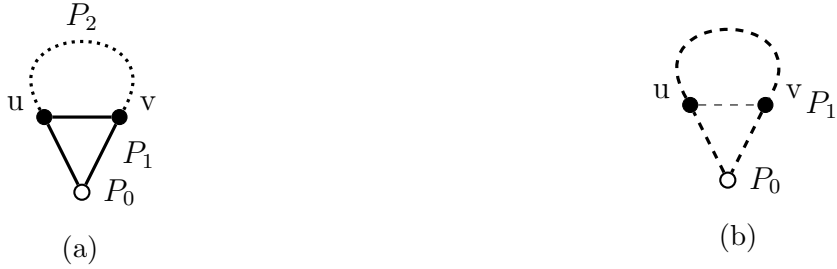
    
    We proceed to make all other short ears pendant, starting with 2-ears. As long as $D$ has a non-pendant 2-ear, we repeat the following procedure. Choose the first non-pendant 2-ear $P$ in $D$. Since $P$ is non-pendant, there exists a nontrivial ear in $D$ with one end incident on the internal vertex $z$ of $P$. Let $Q$ be the first such ear in $D$.
    
    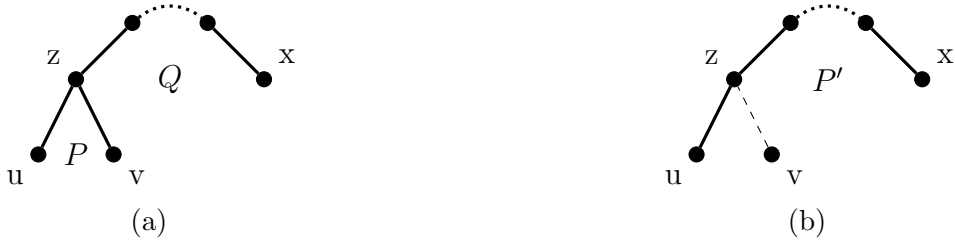
\begin{figure}[ht]
        \begin{subfigure}{0.47\textwidth}
            \centering
            \begin{tikzpicture}[scale=0.5]
            \begin{scope}[every node/.style={circle, fill=black, draw, inner sep=0pt,
            minimum size = 0.2cm
            }]
                \node[label={[label distance=5]225:u}] (u) at (0,0) {};
                \node[label={[label distance=5]315:v}] (v) at (2,0) {};
                \node[label={[label distance=5]135:z}] (z) at (1,2) {};
                \node[label={[label distance=5]45:x}] (x) at (6,2) {};
                \node[] (n1) at (2.5,3.5) {};
                \node[] (n2) at (4.5,3.5) {};
            \end{scope}
            \begin{scope}[every edge/.style={draw=black}]
            \path[very thick] (u) edge node {} (z);
            \path[very thick] (v) edge node {} (z);
            \path[very thick] (z) edge node {} (n1);
            \path[very thick, dotted] (n1) edge[bend left=40] node {} (n2);
            \path[very thick] (x) edge node {} (n2);
            \end{scope}
            \begin{scope}[every node/.style={draw=none,rectangle}]
            \node (p) at (1,0) {$P$};
            \node (q) at (3.5,2) {$Q$};
            \end{scope}
            \end{tikzpicture}
        \caption{}
        \label{twoeara}
        \end{subfigure}
        \hspace*{\fill}
        \begin{subfigure}{0.47\textwidth}
            \centering
            \begin{tikzpicture}[scale=0.5]
            \begin{scope}[every node/.style={circle, fill=black, draw, inner sep=0pt,
            minimum size = 0.2cm
            }]
                \node[label={[label distance=5]225:u}] (u) at (0,0) {};
                \node[label={[label distance=5]315:v}] (v) at (2,0) {};
                \node[label={[label distance=5]135:z}] (z) at (1,2) {};
                \node[label={[label distance=5]45:x}] (x) at (6,2) {};
                \node[] (n1) at (2.5,3.5) {};
                \node[] (n2) at (4.5,3.5) {};
            \end{scope}
            \begin{scope}[every edge/.style={draw=black}]
            \path[very thick] (u) edge node {} (z);
            \path[dashed] (v) edge node {} (z);
            \path[very thick] (z) edge node {} (n1);
            \path[very thick, dotted] (n1) edge[bend left=40] node {} (n2);
            \path[very thick] (x) edge node {} (n2);
            \end{scope}
            \begin{scope}[every node/.style={draw=none,rectangle}]
            \node (p') at (3.5,2) {$P'$};
            \end{scope}
            \end{tikzpicture}
        \caption{}
        \label{twoearb}
        \end{subfigure}
        \caption{$P$ is a 2-ear}
    \end{figure}
    
    Let $u$ and $v$ be the end vertices of $P$ and $x$ and $z$ be the end vertices of $Q$ such that $u \neq x$. Delete ears $P$ and $Q$ from $D$, and construct the ear $P'$ with ends at $u$ and $x$, containing internally the internal vertices of both $P$ and $Q$, as shown by the thick line in Figure \ref{twoearb}. Add the ear $P'$ to $D$ in the position of $Q$. Add the trivial ear $vz$ at the end of $D$. By Lemma \ref{firstearlemma}, $D$ is still a valid ear-decomposition of $G$. Since $Q$ was a nontrivial ear, $P'$ has length at least 3, hence this procedure reduces the number of 2-ears in $D$ by one. Further, $P'$ is an open ear, thus $D$ is still an open ear-decomposition. If $Q$ was an even ear, then this procedure reduced the number of even ears by 2, contradicting our assumption that $D$ is evenmin. Hence $Q$ was an odd ear, and the number of even ears remains unchanged in $D$. 
    
    After repeating the above procedure for all non-pendant 2-ears, all 2-ears in $D$ are pendant. Next, we make all 3-ears pendant. As long as $D$ has a non-pendant 3-ear, we repeat the following procedure. Prior to each iteration, we relabel the ears in $D$ such that the $i^{\text{th}}$ ear is labelled $P_i$.
    
    Let $P$ be the first non-pendant 3-ear in $D$. Let $x$ and $z$ be the endpoints of $P$, and let $v$ and $y$ be the internal vertices of $P$ adjacent to $x$ and $z$ respectively (as shown in Figure \ref{threeear}).
    
    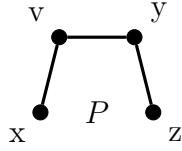
\begin{figure}[ht]
        \centering
        \begin{tikzpicture}[scale=0.5]
        \begin{scope}[every node/.style={circle, fill=black, draw, inner sep=0pt,
        minimum size = 0.2cm
        }]
            \node[label={[label distance=5]225:x}] (x) at (0,0) {};
            \node[label={[label distance=5]315:z}] (z) at (3,0) {};
            \node[label={[label distance=5]135:v}] (v) at (0.5,2) {};
            \node[label={[label distance=5]45:y}] (y) at (2.5,2) {};
        \end{scope}
        \begin{scope}[every edge/.style={draw=black}]
        \path[very thick] (x) edge node {} (v);
        \path[very thick] (v) edge node {} (y);
        \path[very thick] (y) edge node {} (z);
        \end{scope}
        \begin{scope}[every node/.style={draw=none,rectangle}]
        \node (p) at (1.5,0) {$P$};
        \end{scope}
        \end{tikzpicture}
        \caption{$P$ is a 3-ear}
        \label{threeear}
    \end{figure}
    
    \begin{description}
    \item \textbf{Case 1. } There exists a nontrivial ear $Q$ with endpoints $v$ and $y$.
    
    Let $P'$ be the ear with endpoints $x$ and $z$ consisting of all of the edges of $Q$ and the edges $vx$ and $yz$ (as shown by the thick dashed line in Figure \ref{threeearcaseoneb}). Delete the ears $P$ and $Q$ from $D$, and add the ear $P'$ to $D$ in the position of $P$, and the trivial ear $xy$ at the end of $D$. The resulting ear-decomposition $D$ is valid for $G$. Since $Q$ is nontrivial, $P'$ has length at least 4 and is a long ear. Further, $D$ is still an open ear-decomposition, and since the length of $P'$ has the same parity as the length of $Q$, $D$ is still evenmin.
    
    \begin{figure}[ht]
        \begin{subfigure}{0.47\textwidth}
            \centering
            \begin{tikzpicture}[scale=0.5]
            \begin{scope}[every node/.style={circle, fill=black, draw, inner sep=0pt,
            minimum size = 0.2cm
            }]
                \node[label={[label distance=5]225:x}] (x) at (0,0) {};
                \node[label={[label distance=5]315:z}] (z) at (3,0) {};
                \node[label={[label distance=5]165:v}] (v) at (0.5,2) {};
                \node[label={[label distance=5]15:y}] (y) at (2.5,2) {};
            \end{scope}
            \begin{scope}[every edge/.style={draw=black}]
            \path[very thick] (x) edge node {} (v);
            \path[very thick] (v) edge node {} (y);
            \path[very thick] (y) edge node {} (z);
            \path[very thick, dotted] (v) edge[min distance=80, bend left=120] node[] {} (y);
            \end{scope}
            \begin{scope}[every node/.style={draw=none,rectangle}]
            \node (p) at (1.5,0) {$P$};
            \node (q) at (1.5,3) {$Q$};
            \end{scope}
            \end{tikzpicture}
            \caption{}
            \label{threeearcaseonea}
        \end{subfigure}
        \hspace*{\fill}
        \begin{subfigure}{0.47\textwidth}
            \centering
            \begin{tikzpicture}[scale=0.5]
            \begin{scope}[every node/.style={circle, fill=black, draw, inner sep=0pt,
            minimum size = 0.2cm
            }]
                \node[label={[label distance=5]225:x}] (x) at (0,0) {};
                \node[label={[label distance=5]315:z}] (z) at (3,0) {};
                \node[label={[label distance=5]165:v}] (v) at (0.5,2) {};
                \node[label={[label distance=5]15:y}] (y) at (2.5,2) {};
            \end{scope}
            \begin{scope}[every edge/.style={draw=black}]
            \path[very thick, dashed] (x) edge node {} (v);
            \path[dashed] (v) edge node {} (y);
            \path[very thick, dashed] (y) edge node {} (z);
            \path[very thick, dashed] (v) edge[min distance=80, bend left=120] node[] {} (y);
            \end{scope}
            \begin{scope}[every node/.style={draw=none,rectangle}]
            \node (p') at (1.5,1) {$P'$};
            \end{scope}
            \end{tikzpicture}
            \caption{}
            \label{threeearcaseoneb}
        \end{subfigure}
        \caption{There exists a nontrivial ear $Q$ with endpoints $v$ and $y$}
    \end{figure}
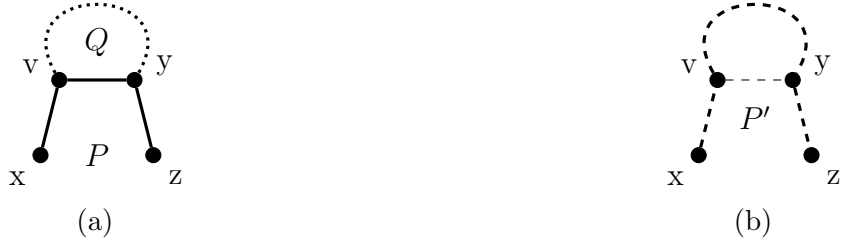
    
    \item \textbf{Case 2. } There exist ears $Q_1$ and $Q_2$ such that $Q_1$ has endpoints $x$ and $y$, $Q_2$ has endpoints $v$ and $z$, and at least one ear in $\{Q_1,Q_2\}$ is nontrivial.
    
    Let $P'$ be the ear with endpoints $x$ and $z$ consisting of all the edges of $Q_1$ and $Q_2$ and the edge $vy$ (as shown by the thick dashed line in Figure \ref{threeearcasetwob}). Delete the ears $P$, $Q_1$ and $Q_2$ from $D$, and add the ear $P'$ to $D$ in the position of $P$, and the trivial ears $ax$ and $by$ at the end of $D$. The resulting open ear-decomposition $D$ is valid for $G$. If both $Q_1$ and $Q_2$ are even ears, then this procedure reduces the number of even ears in $D$ by 2, contradicting our assumption that $D$ was evenmin. If either zero or exactly one of these ears is even, then $D$ remains evenmin.
    
    \begin{figure}[ht]
        \begin{subfigure}{0.47\textwidth}
            \centering
            \begin{tikzpicture}[scale=0.5]
            \begin{scope}[every node/.style={circle, fill=black, draw, inner sep=0pt,
            minimum size = 0.2cm
            }]
                \node[label={[label distance=5]225:x}] (x) at (0,0) {};
                \node[label={[label distance=5]315:z}] (z) at (3,0) {};
                \node[label={[label distance=5]135:v}] (v) at (0.5,2) {};
                \node[label={[label distance=5]45:y}] (y) at (2.5,2) {};
            \end{scope}
            \begin{scope}[every edge/.style={draw=black}]
            \path[very thick] (x) edge node {} (v);
            \path[very thick] (v) edge node {} (y);
            \path[very thick] (y) edge node {} (z);
            \path[very thick, dotted] (x) edge[min distance=180, bend left=120] node[] {} (y);
            \path[very thick, dotted] (v) edge[min distance=180, bend left=120] node[] {} (z);
            \end{scope}
            \begin{scope}[every node/.style={draw=none,rectangle}]
            \node (p) at (1.5,0) {$P$};
            \node (q1) at (-1,3.5) {$Q_1$};
            \node (q2) at (4,3.5) {$Q_2$};
            \end{scope}
            \end{tikzpicture}
            \caption{}
            \label{threeearcasetwoa}
        \end{subfigure}
        \hspace*{\fill}
        \begin{subfigure}{0.47\textwidth}
            \centering
            \begin{tikzpicture}[scale=0.5]
            \begin{scope}[every node/.style={circle, fill=black, draw, inner sep=0pt,
            minimum size = 0.2cm
            }]
                \node[label={[label distance=5]225:x}] (x) at (0,0) {};
                \node[label={[label distance=5]315:z}] (z) at (3,0) {};
                \node[label={[label distance=5]135:v}] (v) at (0.5,2) {};
                \node[label={[label distance=5]45:y}] (y) at (2.5,2) {};
            \end{scope}
            \begin{scope}[every edge/.style={draw=black}]
            \path[dashed] (x) edge node {} (v);
            \path[very thick, dashed] (v) edge node {} (y);
            \path[dashed] (y) edge node {} (z);
            \path[very thick, dashed] (x) edge[min distance=180, bend left=120] node[] {} (y);
            \path[very thick, dashed] (v) edge[min distance=180, bend left=120] node[] {} (z);
            \end{scope}
            \begin{scope}[every node/.style={draw=none,rectangle}]
            \node (p') at (1.5,1) {$P'$};
            \end{scope}
            \end{tikzpicture}
            \caption{}
            \label{threeearcasetwob}
        \end{subfigure}
        \caption{There exist ears $Q_1$ from $x$ to $y$ and $Q_2$ from $v$ to $z$, not both trivial}
    \end{figure}
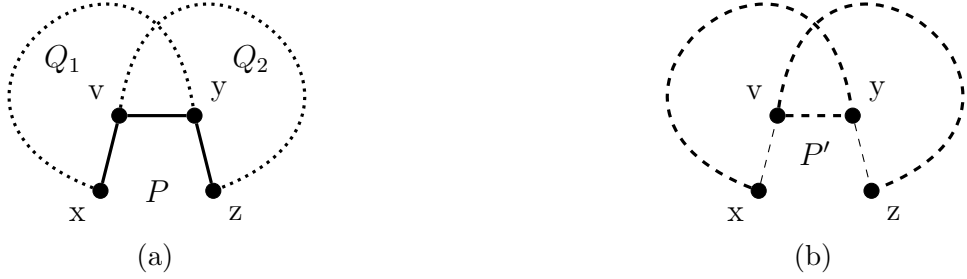
    
    \item \textbf{Case 3. } Otherwise, let $Q$ be the first nontrivial ear with an endpoint at an internal vertex of $P$ (say $y$). Let $w$ be the other endpoint of $Q$ (as shown in Figure \ref{threeearcasethree}).
    
    \begin{figure}[ht]
        \centering
        \begin{tikzpicture}[scale=0.5]
            \begin{scope}[every node/.style={circle, fill=black, draw, inner sep=0pt,
            minimum size = 0.2cm
            }]
                \node[label={[label distance=5]225:x}] (x) at (0,0) {};
                \node[label={[label distance=5]135:v}] (v) at (0.5,2) {};
                \node[label={[label distance=5]135:y}] (y) at (2.5,2) {};
                \node[label={[label distance=5]315:z}] (z) at (3,0) {};
                \node[label={[label distance=5]45:w}] (w) at (7.5,2) {};
                \node[] (n1) at (4,3.5) {};
                \node[] (n2) at (6,3.5) {};
            \end{scope}
            \begin{scope}[every edge/.style={draw=black}]
            \path[very thick] (x) edge node {} (v);
            \path[very thick] (v) edge node {} (y);
            \path[very thick] (y) edge node {} (z);
            \path[very thick] (y) edge node {} (n1);
            \path[very thick, dotted] (n1) edge[bend left=40] node {} (n2);
            \path[very thick] (n2) edge node {} (w);
            \end{scope}
            \begin{scope}[every node/.style={draw=none,rectangle}]
            \node (p) at (1.5,0) {$P$};
            \node (q) at (5,2) {$Q$};
            \end{scope}
        \end{tikzpicture}
    \caption{}
    \label{threeearcasethree}
    \end{figure}
    
    \begin{description}
    \item \textbf{Case 3a. } $w \neq x$. Let $P'$ be the ear with endpoints $x$ and $w$, consisting of the edges of $P$ and $Q$ except $yz$. Delete the ears $P$ and $Q$ from $D$, add the ear $P'$ to $D$ in the position of $Q$, and the trivial ear $yz$ at the end of $D$. $P'$ is both open and long, and the resulting ear-decomposition $D$ is valid for $G$ by Lemma \ref{firstearlemma}. Since the length of $P'$ has the same parity as the length of $Q$, the number of even ears remains the same.
    
    \begin{figure}[ht]
    \begin{subfigure}{0.47\textwidth}
        \centering
        \begin{tikzpicture}[scale=0.5]
            \begin{scope}[every node/.style={circle, fill=black, draw, inner sep=0pt,
            minimum size = 0.2cm
            }]
                \node[label={[label distance=3]225:x}] (x) at (0,0) {};
                \node[label={[label distance=3]135:v}] (v) at (0.5,2) {};
                \node[label={[label distance=3]135:y}] (y) at (2.5,2) {};
                \node[label={[label distance=3]315:z}] (z) at (3,0) {};
                \node[label={[label distance=3]45:u}] (u) at (3,-2) {};
                \node[draw=none,fill=none] (n1) at (1.5,-2.5) {};
                \node[draw=none,fill=none] (n2) at (4.5,-2.5) {};
            \end{scope}
            \begin{scope}[every edge/.style={draw=black}]
            \path[very thick] (x) edge node {} (v);
            \path[very thick] (v) edge node {} (y);
            \path[very thick] (y) edge node {} (z);
            \path[very thick, dotted] (y) edge[bend right=120, min distance=50mm] node {} (x);
            \path[thick, dashed] (v) edge[bend left=20] node {} (u);
            \path[] (n1) edge[bend left=10] node {} (u);
            \path[] (u) edge[bend left=10] node {} (n2);
            \end{scope}
            \begin{scope}[every node/.style={draw=none,rectangle}]
            \node (p) at (1.25,0) {$P$};
            \node (q) at (-2.5,2) {$Q$};
            \node (r) at (3,-3) {$R$};
            \end{scope}
        \end{tikzpicture}
    \caption{}
    \label{threeearcasethreea}
    \end{subfigure}
    \hspace*{\fill} 
    \begin{subfigure}{0.47\textwidth}
        \centering
        \begin{tikzpicture}[scale=0.5]
            \begin{scope}[every node/.style={circle, fill=black, draw, inner sep=0pt,
            minimum size = 0.2cm
            }]
                \node[label={[label distance=5]225:x}] (x) at (0,0) {};
                \node[label={[label distance=5]135:v}] (v) at (0.5,2) {};
                \node[label={[label distance=5]135:y}] (y) at (2.5,2) {};
                \node[label={[label distance=5]315:z}] (z) at (3,0) {};
                \node[label={[label distance=5]270:t}] (n1) at (4,6) {};
                \node[] (n2) at (2,4.5) {};
                \node[] (n3) at (6,4.5) {};
                \node[] (n4) at (4,3) {};
                \node[] (n5) at (8,3) {};
            \end{scope}
            \begin{scope}[every edge/.style={draw=black}]
            \path[very thick] (x) edge node {} (v);
            \path[very thick] (v) edge node {} (y);
            \path[very thick] (y) edge node {} (z);
            \path[very thick, dotted] (y) edge[bend right=120, min distance=50mm] node {} (x);
            \path[very thick, dashed] (v) edge[bend left=60] node {} (n1);
            \path[thick, dotted] (n1) edge[bend right=40] node {} (n2);
            \path[very thick, dashed] (n1) edge[bend left=40] node {} (n3);
            \path[very thick, dashed] (n3) edge[bend right=40] node {} (n4);
            \path[thick, dotted] (n3) edge[bend left=40] node {} (n5);
            \end{scope}
            \begin{scope}[every node/.style={draw=none,rectangle}]
            \node (p) at (1.5,0) {$P$};
            \node (q) at (-2.5,2) {$Q$};
            \node (r) at (0.5,5.5) {$R$};
            \node (l1) at (6,6) {$L_1$};
            \node (l2) at (5.25,3.75) {$L_2$};
            \end{scope}
        \end{tikzpicture}
    \caption{}
    \label{threeearcasethreeb}
    \end{subfigure}
    \caption{}
    \end{figure}
    
    \item \textbf{Case 3b. } $w=x$, and $v$ is the endpoint of a trivial ear $uv$ such that $u \in X$. We refer the reader to Figure \ref{threeearcasethreea} for this case.
    
    Let $R$ be the ear containing $u$ internally. If $R$ is a short ear, then it is pendant (since $P$ is the first non-pendant short ear). We have the following cases:
    \begin{enumerate}[label=(\roman*)]
        \item $R$ is a 2-ear. Choose an endpoint $a$ of $R$ that does not coincide with $z$, and let $P'$ be the ear $au \cup uv \cup vy \cup yz$. Delete $P$ and $R$ from $D$ and add the 4-ear $P'$ in the position of $P$, and the new trivial ears at the end of $D$.
        \item $R$ is a 3-ear. Choose the endpoint $a$ of $R$ which is not adjacent to $u$ in $R$. Let $R'$ be the ear of length 2 in $R$ with endpoints $a$ and $u$. If $a$ does not coincide with $x$, let $P'$ be the ear $R' \cup uv \cup vy \cup Q$, which has the same parity as $Q$. Delete $R$, $P$ and $Q$ from $D$ and add $P'$ in the position of $P$ in $D$ and the trivial ears at the end of $D$. If $a$ coincides with $x$, let $P'$ be the ear $R' \cup uv \cup vy \cup yz$ of length 5. Delete $R$ and $P$ from $D$ and add $P'$ in the position of $P$ in $D$ and the trivial ears at the end of $D$.
    \end{enumerate}
     In both of the above cases, we do not create extra even ears, and the resulting ear-decomposition is open, evenmin, and valid for $G$.
     
     If $R$ is a long ear, let $P'$ be the ear $Q \cup yv \cup vu$. Delete $P$ and $Q$ from $D$ and add $P'$ in the position of $P$ in $D$, and the trivial ears at the end of $D$. This ear has the same parity as the ear $Q$, so the resulting ear-decomposition is open, evenmin, and valid for $G$.
    
    \item \textbf{Case 3c. } Otherwise, since the graph has minimum degree at least 3, $v$ is adjacent to a vertex $u \notin X\cup\{y\}$. Observe that this is the only remaining case.
    
    Let $R$ be the ear containing the edge $uv$, and let $t$ be the other endpoint of $R$. In particular, if $uv$ is a trivial ear, then $t$ is the vertex $u$. We refer the reader to Figure \ref{threeearcasethreeb}, which will be useful throughout the following analysis.
    
    The following sub-procedure constructs three sets of ears ($\myN$, $\myL$ and $\myM$), which are later used to modify the ear-decomposition in order to add the internal vertices of $P$ to a new long ear. The procedure adds some of the existing ears of $D$ to the set $\myN$, and constructs sets of new ears $\myL$ and $\myM$. When suitable sets are found, the ears in $\myN$ are deleted from $D$ and replaced with the ears in $\myM$, along with a suitably constructed long ear. 

    Repeat the following sub-procedure until $t$ is in $X \cup \{v,y\}$. Initialize $\myN$, $\myL$ and $\myM$ with the empty set. Let $S$ be the ear that internally contains $t$, with endpoints $c$ and $d$. Add $S$ to $\myN$. Partition the edges of $S$ into the ears $\mySa$ (with endpoints $c$ and $t$) and $\mySb$ (with endpoints $t$ and $d$). If $S$ is an even ear, either $\mySa$ and $\mySb$ are both odd or they are both even. If $S$ is an odd ear, suppose without loss of generality that $\mySa$ is even and $\mySb$ is odd. Add $\mySa$ to $\myL$ and $\mySb$ to $\myM$. Set $t = c$.
    
    The following observations will be useful in our analysis.
    \begin{itemize}
        \item If $\mySb$ is even, then $S$ is even (thus replacing $S$ with $\mySb$ in any ear-decomposition will not, by itself, increase the number of even ears in that ear-decomposition).
        \item If $\mySa$ is odd, then $\mySb$ is odd and $S$ is even.
    \end{itemize}
    
    When this sub-procedure terminates, we have the following cases:
    \begin{enumerate}[label=(\roman*)]
        \item $t \notin \{v,y,z\}$. Let $P'$ be the ear $zy \cup yv \cup R \cup L_1 \cup \ldots \cup L_k$, where $\myL = \{L_1,\ldots,L_k\}$. Delete $P$ and all ears in $\myN$ from $D$, and for each ear in $\myN$, replace it with the corresponding ear in $\myM$ at the same position in $D$ (this does not add any extra even ears, but might create new non-pendant short ears; observe that these ears occur later in the ear-decomposition than the newly created long ear in this iteration). In the position of $P$, add the ear $P'$. Add the trivial ear $xv$ at the end of $D$. The resulting ear decomposition is valid because the sub-procedure is terminated when a vertex in $X$ is encountered, thus every ear in $\myN$ appeared after $P$ in $D$, hence every ear in $\myM$ appears after $P'$.
        
        If $\myL$ contains only even ears, then $P'$ has the same parity as $R$ and we do not introduce any extra even ears. If not, then $\myL$ contains at least one odd ear, in which case the corresponding ear in $\myN$ is even, and the corresponding ear in $\myM$ is odd. Since we have already reduced the number of even ears by at least one, $P'$ is even and we do not introduce extra even ears.
        \item $t = v$. Discard the previous sets $\myN$, $\myL$ and $\myM$. Choose $u$ to be the neighbour of $v$ on the last ear that was labelled $S$. Let $R$ be this ear and let $t$ be its other endpoint. Since every choice of $R$ that we make in this manner appears strictly earlier in the ear decomposition than all the previous choices, the sub-procedure can only be repeated $O(n)$ times before we no longer have this case.
        \item $t = y$. Let $\myL = \{L_1,\ldots,L_k\}$, and let $P'$ be the ear $xv \cup R \cup L_1 \cup \ldots \cup L_k \cup yz$. Delete $P$ and all ears in $\myN$ from $D$, and for each ear in $\myN$, replace it with the corresponding ear in $\myM$ at the same position in $D$. In the position of $P$, add the ear $P'$. Add the trivial ear $vy$ at the end of $D$. This ear-decomposition is valid, as explained earlier.
        
        If all of the ears in $L$ are even, then $P'$ has the same parity as $R$, and this step does not introduce any extra even ears. If not, then $L$ contains at least one odd ear, in which case $P'$ is even and we do not introduce extra even ears (as explained earlier).
        \item $t = z$. Let $P'$ be the ear $Q \cup vy \cup R \cup L_1 \cup \ldots \cup L_k$, where $\myL = \{L_1,\ldots,L_k\}$. Delete $P$, and in the position of $P$, add the ear $P'$. Delete all ears in $\myN$ from $D$, and for each ear in $\myN$, replace it with the corresponding ear in $\myM$ at the same position in $D$. Add the trivial ears $xv$ and $yz$ at the end of $D$. As explained earlier, this ear-decomposition is valid. If all the ears in $\myL$ are even, we have the following cases:
        \begin{enumerate}[label=(\alph*)]
            \item $Q$ and $R$ are odd. In this case, $P'$ is odd and we do not introduce extra even ears.
            \item Exactly one of $Q$ and $R$ is even. In this case, $P'$ is even and we do not introduce extra even ears.
            \item $Q$ and $R$ are even. In this case, $P'$ is odd, contradicting the assumption that $D$ was evenmin; this case cannot occur.
        \end{enumerate}
        If $\myL$ contains an odd ear, then the corresponding ear in $\myM$ is odd and the corresponding ear in $\myN$ is even. Since we have already reduced the number of even ears by at least one, we do not introduce extra even ears.
        
    \end{enumerate}
    
    \end{description}
    
    \end{description}
    
    In all of the above cases, the internal vertices of $P$ are added to a long ear in $D$. If $P$ was a 3-ear, then it is possible that we created new non-pendant short ears that appear after $P'$ in the new ear-decomposition. These short ears are handled in future iterations of the above procedure, in the same manner as above (that is, first we handle all non-pendant 2-ears, then we handle the first non-pendant 3-ear).
    
    In each iteration, the above procedure takes time polynomial in $|V(G)|$ for the non-pendant short ear under consideration. Further, if this short ear is a 3-ear, the internal vertices of this short ear are added to long ears, and are never again added to a non-pendant short ear until termination. As a consequence, the set $X$ in each iteration is a strict superset of the corresponding set in any previous iteration. Hence the running time for the whole procedure is polynomial. 
    
    On termination of this procedure, the ear-decomposition $D$ is open and has $\phi(G)$ even ears, and all of its short ears are pendant.
\end{proof}

\begin{thm} \label{opennicetheorem}
    Given a 2-vertex-connected graph $G$ with minimum degree at least 3, and an associated evenmin ear-decomposition $D$ in which all short ears are pendant, an open evenmin nice ear-decomposition of $G$ can be computed in polynomial time.
\end{thm}
\begin{proof}
    Since $D$ is open and evenmin, and all short ears of $D$ are pendant, it remains to obtain the property that there are no edges connecting an internal vertex of one short ear to an internal vertex of another short ear of $D$.
    
    Since $D$ has $\phi(G)$ even ears, there are no edges connecting the internal vertices of 2-ears. If not, we could replace the 2-ears and the trivial ear connecting their internal vertices by a pendant 3-ear and two trivial ears, reducing the number of even ears by two, contradicting the assumption that $D$ is evenmin. Since we have two choices for each end vertex of such a 3-ear, we can always choose its end vertices such that it is open.
    
    As long as $D$ has two short pendant ears $P'$ and $P''$ with an edge $e$ connecting an internal vertex of $P'$ with an internal vertex of $P''$, we repeat the following procedure.
    
    \begin{description}
    \item \textbf{Case 1. } One of the ears $P'$ and $P''$ is a 2-ear.
    
    Without loss of generality, assume $P'$ is a 2-ear and $P''$ is a 3-ear.
    
    Let $a$ and $b$ be the endpoints of $P'$ and $z$ be the internal vertex of $P'$. Let $c$ and $d$ be the endpoints of $P''$ and $x$ and $y$ be the internal vertices of $P''$ such that $x$ is adjacent to both $c$ and $z$ (Figure \ref{case1a}). Construct the ear $S$ as shown by the thick paths in Figures \ref{case1b} and \ref{case1c}, that is, $S$ consists of the edges $az$, $zx$, $xy$ and $yd$ if the vertices $a$ and $d$ are distinct, and the edges $bz$, $zx$, $xy$ and $yd$ if they coincide. 
    
    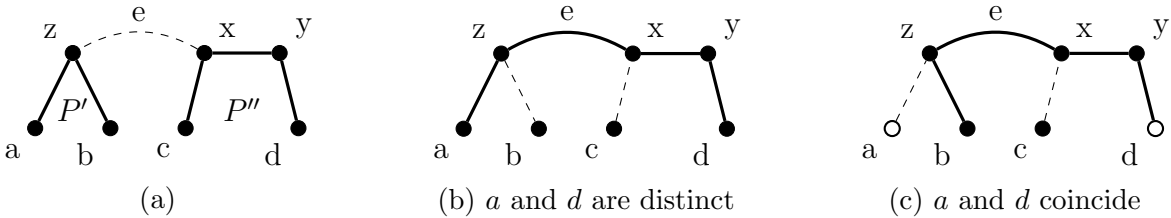
\begin{figure}[ht]
    \begin{subfigure}{0.31\textwidth}
        \centering
        \begin{tikzpicture}[scale=0.5]
        \begin{scope}[every node/.style={circle, fill=black, draw, inner sep=0pt,
        minimum size = 0.2cm
        }]
            \node[label={[label distance=5]225:a}] (a) at (0,0) {};
            \node[label={[label distance=5]225:b}] (b) at (2,0) {};
            \node[label={[label distance=5]225:c}] (c) at (4,0) {};
            \node[label={[label distance=5]225:d}] (d) at (7,0) {};
            \node[label={[label distance=5]135:z}] (z) at (1,2) {};
            \node[label={[label distance=5]45:x}] (x) at (4.5,2) {};
            \node[label={[label distance=5]45:y}] (y) at (6.5,2) {};
        \end{scope}
        \begin{scope}[every edge/.style={draw=black}]
        \path[very thick] (a) edge node {} (z);
        \path[very thick] (z) edge node {} (b);
        \path[very thick] (c) edge node {} (x);
        \path[very thick] (x) edge node {} (y);
        \path[very thick] (y) edge node {} (d);
        \path[dashed] (x) edge[bend right=30] node[label={[label distance=0]90:e}] {} (z);
        \end{scope}
        \begin{scope}[every node/.style={draw=none,rectangle}]
            \node (p') at (1,0.5) {$P'$};
            \node (p'') at (5.5,0.5) {$P''$};
        \end{scope}
        \end{tikzpicture}
    \caption{}
    \label{case1a}
    \end{subfigure}
    \hspace*{\fill} 
    \begin{subfigure}{0.31\textwidth}
        \centering
        \begin{tikzpicture}[scale=0.5]
        \begin{scope}[every node/.style={circle, fill=black, draw, inner sep=0pt,
        minimum size = 0.2cm
        }]
            \node[label={[label distance=5]225:a}] (a) at (0,0) {};
            \node[label={[label distance=5]225:b}] (b) at (2,0) {};
            \node[label={[label distance=5]225:c}] (c) at (4,0) {};
            \node[label={[label distance=5]225:d}] (d) at (7,0) {};
            \node[label={[label distance=5]135:z}] (z) at (1,2) {};
            \node[label={[label distance=5]45:x}] (x) at (4.5,2) {};
            \node[label={[label distance=5]45:y}] (y) at (6.5,2) {};
        \end{scope}
        \begin{scope}[every edge/.style={draw=black}]
        \path[very thick] (a) edge node {} (z);
        \path[dashed] (z) edge node {} (b);
        \path[dashed] (c) edge node {} (x);
        \path[very thick] (x) edge node {} (y);
        \path[very thick] (y) edge node {} (d);
        \path[very thick] (x) edge[bend right=30] node[label={[label distance=0]90:e}] {} (z);
        \end{scope}
        \end{tikzpicture}
    \caption{$a$ and $d$ are distinct}
    \label{case1b}
    \end{subfigure}
    \hspace*{\fill} 
    \begin{subfigure}{0.31\textwidth}
        \centering
        \begin{tikzpicture}[scale=0.5]
        \begin{scope}[every node/.style={circle, fill=black, draw, inner sep=0pt,
        minimum size = 0.2cm
        }]
            \node[fill=none, thick, label={[label distance=5]225:a}] (a) at (0,0) {};
            \node[label={[label distance=5]225:b}] (b) at (2,0) {};
            \node[label={[label distance=5]225:c}] (c) at (4,0) {};
            \node[fill=none, thick, label={[label distance=5]225:d}] (d) at (7,0) {};
            \node[label={[label distance=5]135:z}] (z) at (1,2) {};
            \node[label={[label distance=5]45:x}] (x) at (4.5,2) {};
            \node[label={[label distance=5]45:y}] (y) at (6.5,2) {};
        \end{scope}
        \begin{scope}[every edge/.style={draw=black}]
        \path[dashed] (a) edge node {} (z);
        \path[very thick] (z) edge node {} (b);
        \path[dashed] (c) edge node {} (x);
        \path[very thick] (x) edge node {} (y);
        \path[very thick] (y) edge node {} (d);
        \path[very thick] (x) edge[bend right=30] node[label={[label distance=0]90:e}] {} (z);
        \end{scope}
        \end{tikzpicture}
    \caption{$a$ and $d$ coincide}
    \label{case1c}
    \end{subfigure}
    \caption{$P'$ is a 2-ear} \label{case1}
    \end{figure}
    
    Remove the ears $P'$ and $P''$ from $D$, and add the ear $S$ in place of the ear $P'$, followed by trivial ears consisting of the remaining edges from $P'$ and $P''$ that are not in $S$. Since $P'$ and $P''$ are both pendant ears, the new ear-decomposition is a valid ear-decomposition of $G$. Since the end vertices of $S$ are distinct, it is open, and since we deleted a 2-ear from $D$ before adding a 4-ear to it, the number of even ears in $D$ remains equal to $\phi(G)$.
    
    \item \textbf{Case 2. } Both $P'$ and $P''$ are 3-ears.
    
    Let  $a$ and $b$ be the endpoints of $P'$ and let $v$ and $w$ be its internal vertices adjacent to $a$ and $b$ respectively. Let  $c$ and $d$ be the endpoints of $P''$ and let $x$ and $y$ be its internal vertices adjacent to $c$ and $d$ respectively (Figure \ref{case2}). Suppose $v$ and $y$ are adjacent. We have the following cases.
    
    \begin{figure}[ht]
    \centering
        \begin{tikzpicture}[scale=0.5]
        \begin{scope}[every node/.style={circle, fill=black, draw, inner sep=0pt,
        minimum size = 0.2cm
        }]
            \node[label={[label distance=5]225:a}] (a) at (0,0) {};
            \node[label={[label distance=5]225:b}] (b) at (3,0) {};
            \node[label={[label distance=5]225:c}] (c) at (5,0) {};
            \node[label={[label distance=5]225:d}] (d) at (8,0) {};
            \node[label={[label distance=5]135:v}] (v) at (0.5,2) {};
            \node[label={[label distance=5]45:w}] (w) at (2.5,2) {};
            \node[label={[label distance=5]135:x}] (x) at (5.5,2) {};
            \node[label={[label distance=5]45:y}] (y) at (7.5,2) {};
        \end{scope}
        \begin{scope}[every edge/.style={draw=black}]
        \path[very thick] (a) edge node {} (v);
        \path[very thick] (v) edge node {} (w);
        \path[very thick] (w) edge node {} (b);
        \path[very thick] (c) edge node {} (x);
        \path[very thick] (x) edge node {} (y);
        \path[very thick] (y) edge node {} (d);
        \path[dashed] (v) edge[bend left=35] node[label={[label distance=0]90:e}] {} (y);
        \end{scope}
        \begin{scope}[every node/.style={draw=none,rectangle}]
            \node (p') at (1.5,0.5) {$P'$};
            \node (p'') at (6.5,0.5) {$P''$};
        \end{scope}
        \end{tikzpicture}
    \caption{Both $P'$ and $P''$ are 3-ears}
    \label{case2}
    \end{figure}
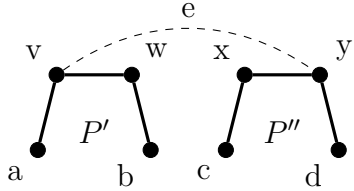
    
    \begin{description}
        \item \textbf{Case 2a. } The vertices $b$ and $c$ are distinct.
        
        Construct the ear $S$ with endpoints $b$ and $c$ and edges $bw$, $wv$, $vy$, $yx$ and $xc$ (as shown by the thick path in Figure \ref{case2a}). Remove the ears $P'$ and $P''$ from $D$, add the ear $S$ in place of the ear $P'$, and add the trivial ears consisting of the remaining edges from $P'$ and $P''$ that are not in $S$ at the end of $D$. Since $P'$ and $P''$ are both pendant ears, the  new ear-decomposition is a valid ear-decomposition of $G$. Since the end vertices of $S$ are distinct, it is open, and since $S$ is an odd ear, the number of even ears in $D$ remains equal to $\phi(G)$.
    
        \begin{figure}[ht]
        \centering
            \begin{tikzpicture}[scale=0.5]
            \begin{scope}[every node/.style={circle, fill=black, draw, inner sep=0pt,
            minimum size = 0.2cm
            }]
                \node[label={[label distance=5]225:a}] (a) at (0,0) {};
                \node[label={[label distance=5]225:b}] (b) at (3,0) {};
                \node[label={[label distance=5]225:c}] (c) at (5,0) {};
                \node[label={[label distance=5]225:d}] (d) at (8,0) {};
                \node[label={[label distance=5]135:v}] (v) at (0.5,2) {};
                \node[label={[label distance=5]45:w}] (w) at (2.5,2) {};
                \node[label={[label distance=5]135:x}] (x) at (5.5,2) {};
                \node[label={[label distance=5]45:y}] (y) at (7.5,2) {};
            \end{scope}
            \begin{scope}[every edge/.style={draw=black}]
            \path[dashed] (a) edge node {} (v);
            \path[very thick] (v) edge node {} (w);
            \path[very thick] (w) edge node {} (b);
            \path[very thick] (c) edge node {} (x);
            \path[very thick] (x) edge node {} (y);
            \path[dashed] (y) edge node {} (d);
            \path[very thick] (v) edge[bend left=35] node[label={[label distance=0]90:e}] {} (y);
            \end{scope}
            \end{tikzpicture}
        \caption{$b$ and $c$ are distinct}
        \label{case2a}
        \end{figure}
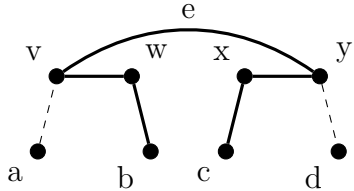
        
        \item \textbf{Case 2b. } The vertices $b$ and $c$ coincide, as shown in Figure \ref{case2b}.
        
        Since every vertex of the graph has degree at least 3, $x$ is adjacent to some vertex not in the set $\{b, y\}$.
    
        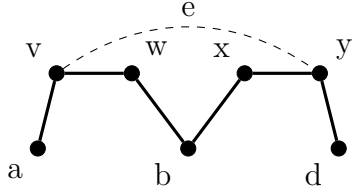
\begin{figure}[ht]
        \centering
            \begin{tikzpicture}[scale=0.5]
            \begin{scope}[every node/.style={circle, fill=black, draw, inner sep=0pt,
            minimum size = 0.2cm
            }]
                \node[label={[label distance=5]225:a}] (a) at (0,0) {};
                \node[label={[label distance=5]225:b}] (b) at (4,0) {};
                \node[label={[label distance=5]225:d}] (d) at (8,0) {};
                \node[label={[label distance=5]135:v}] (v) at (0.5,2) {};
                \node[label={[label distance=5]45:w}] (w) at (2.5,2) {};
                \node[label={[label distance=5]135:x}] (x) at (5.5,2) {};
                \node[label={[label distance=5]45:y}] (y) at (7.5,2) {};
            \end{scope}
            \begin{scope}[every edge/.style={draw=black}]
            \path[very thick] (a) edge node {} (v);
            \path[very thick] (v) edge node {} (w);
            \path[very thick] (w) edge node {} (b);
            \path[very thick] (b) edge node {} (x);
            \path[very thick] (x) edge node {} (y);
            \path[very thick] (y) edge node {} (d);
            \path[dashed] (v) edge[bend left=35] node[label={[label distance=0]90:e}] {} (y);
            \end{scope}
            \end{tikzpicture}
        \caption{$b$ and $c$ coincide}
        \label{case2b}
        \end{figure}
        
        \begin{description}
        \item \textbf{Case 2b.I. } $x$ is adjacent to an internal vertex of $P'$.
        
        If $x$ is adjacent to $v$, construct the ear $S$ with endpoints $b$ and $d$ and edges $bw$, $wv$, $vx$, $xy$ and $yd$ (as shown by the thick path in Figure \ref{case2bi-a}). Otherwise, if $x$ is adjacent to $w$, construct the ear $S$ with endpoints $a$ and $b$ and edges $av$, $vy$, $yx$, $xw$ and $wb$ (as shown by the thick path in Figure \ref{case2bi-b}). In either case, delete $P'$ and $P''$ from $D$, add $S$ to $D$ in place of $P'$, and add all of the remaining edges (dashed edges in the corresponding figure) in trivial ears at the end of $D$.
        
        In both cases, the ear $S$ is an odd long ear with distinct end points, hence the new ear-decomposition is open and is valid for $G$, and the number of even ears remains equal to $\phi(G)$.
        
        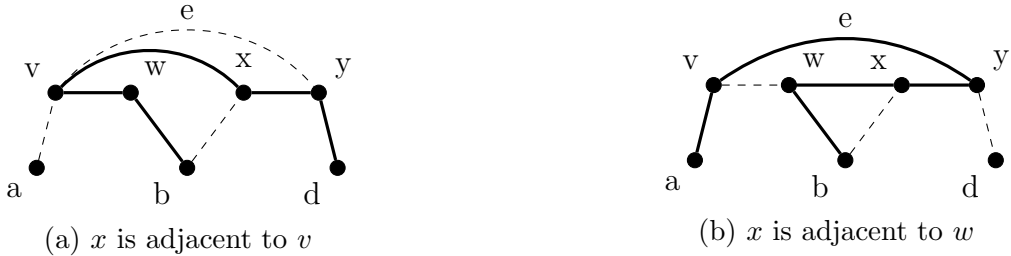
\begin{figure}[ht]
        \begin{subfigure}{0.47\textwidth}
            \centering
            \begin{tikzpicture}[scale=0.5]
            \begin{scope}[every node/.style={circle, fill=black, draw, inner sep=0pt,
            minimum size = 0.2cm
            }]
                \node[label={[label distance=5]225:a}] (a) at (0,0) {};
                \node[label={[label distance=5]225:b}] (b) at (4,0) {};
                \node[label={[label distance=5]225:d}] (d) at (8,0) {};
                \node[label={[label distance=5]135:v}] (v) at (0.5,2) {};
                \node[label={[label distance=5]45:w}] (w) at (2.5,2) {};
                \node[label={[label distance=5]90:x}] (x) at (5.5,2) {};
                \node[label={[label distance=5]45:y}] (y) at (7.5,2) {};
            \end{scope}
            \begin{scope}[every edge/.style={draw=black}]
            \path[dashed] (a) edge node {} (v);
            \path[very thick] (v) edge node {} (w);
            \path[very thick] (w) edge node {} (b);
            \path[dashed] (b) edge node {} (x);
            \path[very thick] (x) edge node {} (y);
            \path[very thick] (y) edge node {} (d);
            \path[dashed] (v) edge[bend left=50] node[label={[label distance=0]90:e}] {} (y);
            \path[very thick] (v) edge[bend left=45] node {} (x);
            \end{scope}
            \end{tikzpicture}
        \caption{$x$ is adjacent to $v$}
        \label{case2bi-a}
        \end{subfigure}
        \hspace*{\fill}
        \begin{subfigure}{0.47\textwidth}
            \centering
            \begin{tikzpicture}[scale=0.5]
            \begin{scope}[every node/.style={circle, fill=black, draw, inner sep=0pt,
            minimum size = 0.2cm
            }]
                \node[label={[label distance=5]225:a}] (a) at (0,0) {};
                \node[label={[label distance=5]225:b}] (b) at (4,0) {};
                \node[label={[label distance=5]225:d}] (d) at (8,0) {};
                \node[label={[label distance=5]135:v}] (v) at (0.5,2) {};
                \node[label={[label distance=5]45:w}] (w) at (2.5,2) {};
                \node[label={[label distance=5]135:x}] (x) at (5.5,2) {};
                \node[label={[label distance=5]45:y}] (y) at (7.5,2) {};
            \end{scope}
            \begin{scope}[every edge/.style={draw=black}]
            \path[very thick] (a) edge node {} (v);
            \path[dashed] (v) edge node {} (w);
            \path[very thick] (w) edge node {} (b);
            \path[dashed] (b) edge node {} (x);
            \path[very thick] (x) edge node {} (y);
            \path[dashed] (y) edge node {} (d);
            \path[very thick] (v) edge[bend left=35] node[label={[label distance=0]90:e}] {} (y);
            \path[very thick] (w) edge node {} (x);
            \end{scope}
            \end{tikzpicture}
        \caption{$x$ is adjacent to $w$}
        \label{case2bi-b}
        \end{subfigure}
        \caption{$x$ is adjacent to an internal vertex of $P'$}
        \label{case2bi}
        \end{figure}
        
        \item \textbf{Case 2b.II. } $x$ is adjacent to an internal vertex $z$ of an ear $R$ not equal to $P'$.
        
        Since the input graph is simple and does not have parallel edges, $z$ does not coincide with $b$ or $y$. If $R$ is a long ear, construct the ear $S$ with endpoints $b$ and $z$ and edges $bw$, $wv$, $vy$, $yx$ and $xz$ (as shown by the thick path in Figure \ref{case2bii-a}). Delete $P'$ and $P''$ from $D$, add $S$ to $D$ in place of $P'$, and add all of the dashed edges in the corresponding figure in trivial ears at the end of $D$.
        
        Otherwise, if $R$ is a short ear, then it is pendant. If it is a 2-ear (Figure \ref{case2bii-b}), construct the ear $S$ as shown by the thick path in the figure. Observe that we have two choices for one end of this ear: we choose to end the ear at either $g$ or $h$, so as to ensure that it is an open ear. The example in the figure shows $S$ ending at $g$, with edges $gz$, $zx$, $xy$, $yv$, $vw$ and $wb$. Remove $P'$, $P''$ and $R$ from $D$, add $S$ to $D$ in place of $P'$, and add all of the dashed edges in trivial ears at the end of $D$.
        
        Otherwise, $R$ is a 3-ear. Let $g$ and $h$ be the endpoints of $R$, and $i$ and $z$ be its internal vertices adjacent to $g$ and $h$ respectively (Figures \ref{case2bii-c} and \ref{case2bii-d}). We have two cases: either $g$ and $b$ are distinct, or they coincide. If they are distinct, construct the ear $S$ as shown by the thick path in Figure \ref{case2bii-c}, with edges $bw$, $wv$, $vy$, $yx$, $xz$, $zi$ and $ig$. Delete $P'$, $P''$ and $R$ from $D$, and add $S$ to $D$ in place of $P'$, and all of the dashed edges in trivial ears at the end of $D$.
        
        Otherwise, the vertices $g$ and $b$ coincide. Construct the ear $S$ as shown by the thick path in Figure \ref{case2bii-d}, with edges $gi$, $iz$, $zx$, $xy$ and $yd$. Delete the ears $P''$ and $R$ from $D$ and add $S$ to $D$ in place of $P''$, and all of the dashed edges in trivial ears at the end of $D$.
        
        In all cases, the number of even ears remains equal to $\phi(G)$, since the only case where $S$ is an even ear is when $R$ is a 2-ear. Additionally, $S$ is an open pendant ear, hence the new ear-decomposition is valid for $G$.
        
        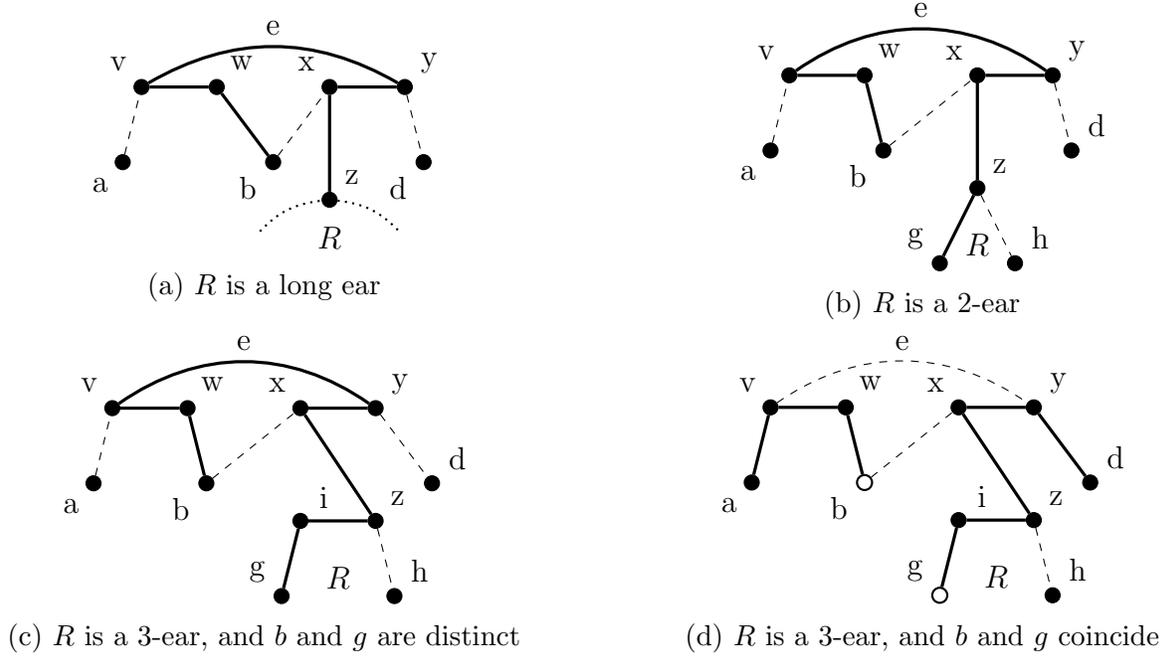
\begin{figure}[ht]
        \begin{subfigure}{0.47\textwidth}
            \centering
            \begin{tikzpicture}[scale=0.5]
            \begin{scope}[every node/.style={circle, fill=black, draw, inner sep=0pt,
            minimum size = 0.2cm
            }]
                \node[label={[label distance=5]225:a}] (a) at (0,0) {};
                \node[label={[label distance=5]225:b}] (b) at (4,0) {};
                \node[label={[label distance=5]225:d}] (d) at (8,0) {};
                \node[label={[label distance=5]135:v}] (v) at (0.5,2) {};
                \node[label={[label distance=5]45:w}] (w) at (2.5,2) {};
                \node[label={[label distance=5]135:x}] (x) at (5.5,2) {};
                \node[label={[label distance=5]45:y}] (y) at (7.5,2) {};
                \node[opacity=0] (left) at (3.5,-2) {};
                \node[label={[label distance=5]45:z},label={[label distance=5]270:$R$}] (mid) at (5.5,-1) {};
                \node[opacity=0] (right) at (7.5,-2) {};
            \end{scope}
            \begin{scope}[every edge/.style={draw=black}]
            \path[dashed] (a) edge node {} (v);
            \path[very thick] (v) edge node {} (w);
            \path[very thick] (w) edge node {} (b);
            \path[dashed] (b) edge node {} (x);
            \path[very thick] (x) edge node {} (y);
            \path[dashed] (y) edge node {} (d);
            \path[very thick] (v) edge[bend left=30]
            node[label={[label distance=0]90:e}] {} (y);
            \path[thick, dotted] (left) edge[bend left=20]
            node {} (mid);
            \path[thick, dotted] (mid) edge[bend left=20]
            node {} (right);
            \path[very thick] (x) edge node {} (mid);
            \end{scope}
            \end{tikzpicture}
        \caption{$R$ is a long ear}
        \label{case2bii-a}
        \end{subfigure}
        \hspace*{\fill}
        \begin{subfigure}{0.47\textwidth}
            \centering
            \begin{tikzpicture}[scale=0.5]
            \begin{scope}[every node/.style={circle, fill=black, draw, inner sep=0pt,
            minimum size = 0.2cm
            }]
                \node[label={[label distance=5]225:a}] (a) at (0,0) {};
                \node[label={[label distance=5]225:b}] (b) at (3,0) {};
                \node[label={[label distance=5]45:d}] (d) at (8,0) {};
                \node[label={[label distance=5]135:v}] (v) at (0.5,2) {};
                \node[label={[label distance=5]45:w}] (w) at (2.5,2) {};
                \node[label={[label distance=5]135:x}] (x) at (5.5,2) {};
                \node[label={[label distance=5]45:y}] (y) at (7.5,2) {};
                \node[label={[label distance=5]45:z}] (z) at (5.5,-1) {};
                \node[label={[label distance=5]135:g}] (g) at (4.5,-3) {};
                \node[label={[label distance=5]45:h}] (h) at (6.5,-3) {};
            \end{scope}
            \begin{scope}[every edge/.style={draw=black}]
            \path[dashed] (a) edge node {} (v);
            \path[very thick] (v) edge node {} (w);
            \path[very thick] (w) edge node {} (b);
            \path[dashed] (b) edge node {} (x);
            \path[very thick] (x) edge node {} (y);
            \path[dashed] (y) edge node {} (d);
            \path[very thick] (v) edge[bend left=35] node[label={[label distance=0]90:e}] {} (y);
            \path[very thick] (x) edge node {} (z);
            \path[very thick] (z) edge node {} (g);
            \path[dashed] (z) edge node {} (h);
            \end{scope}
            \begin{scope}[every node/.style={draw=none,rectangle}]
            \node (r) at (5.5,-2.5) {$R$};
            \end{scope}
            \end{tikzpicture}
        \caption{$R$ is a 2-ear}
        \label{case2bii-b}
        \end{subfigure}
        
        \begin{subfigure}{0.47\textwidth}
            \centering
            \begin{tikzpicture}[scale=0.5]
            \begin{scope}[every node/.style={circle, fill=black, draw, inner sep=0pt,
            minimum size = 0.2cm
            }]
                \node[label={[label distance=5]225:a}] (a) at (0,0) {};
                \node[label={[label distance=5]225:b}] (b) at (3,0) {};
                \node[label={[label distance=5]45:d}] (d) at (9,0) {};
                \node[label={[label distance=5]135:v}] (v) at (0.5,2) {};
                \node[label={[label distance=5]45:w}] (w) at (2.5,2) {};
                \node[label={[label distance=5]135:x}] (x) at (5.5,2) {};
                \node[label={[label distance=5]45:y}] (y) at (7.5,2) {};
                \node[label={[label distance=5]45:z}] (z) at (7.5,-1) {};
                \node[label={[label distance=5]135:g}] (g) at (5,-3) {};
                \node[label={[label distance=5]45:h}] (h) at (8,-3) {};
                \node[label={[label distance=5]45:i}] (i) at (5.5,-1) {};
            \end{scope}
            \begin{scope}[every edge/.style={draw=black}]
            \path[dashed] (a) edge node {} (v);
            \path[very thick] (v) edge node {} (w);
            \path[very thick] (w) edge node {} (b);
            \path[dashed] (b) edge node {} (x);
            \path[very thick] (x) edge node {} (y);
            \path[dashed] (y) edge node {} (d);
            \path[very thick] (v) edge[bend left=35] node[label={[label distance=0]90:e}] {} (y);
            \path[very thick] (x) edge node {} (z);
            \path[very thick] (z) edge node {} (i);
            \path[very thick] (i) edge node {} (g);
            \path[dashed] (z) edge node {} (h);
            \end{scope}
            \begin{scope}[every node/.style={draw=none,rectangle}]
            \node (r) at (6.5,-2.5) {$R$};
            \end{scope}
            \end{tikzpicture}
        \caption{$R$ is a 3-ear, and $b$ and $g$ are distinct}
        \label{case2bii-c}
        \end{subfigure}
        \hspace*{\fill}
        \begin{subfigure}{0.47\textwidth}
            \centering
            \begin{tikzpicture}[scale=0.5]
            \begin{scope}[every node/.style={circle, fill=black, draw, inner sep=0pt,
            minimum size = 0.2cm
            }]
                \node[label={[label distance=5]225:a}] (a) at (0,0) {};
                \node[fill=none, thick, label={[label distance=5]225:b}] (b) at (3,0) {};
                \node[label={[label distance=5]45:d}] (d) at (9,0) {};
                \node[label={[label distance=5]135:v}] (v) at (0.5,2) {};
                \node[label={[label distance=5]45:w}] (w) at (2.5,2) {};
                \node[label={[label distance=5]135:x}] (x) at (5.5,2) {};
                \node[label={[label distance=5]45:y}] (y) at (7.5,2) {};
                \node[label={[label distance=5]45:z}] (z) at (7.5,-1) {};
                \node[fill=none, thick, label={[label distance=5]135:g}] (g) at (5,-3) {};
                \node[label={[label distance=5]45:h}] (h) at (8,-3) {};
                \node[label={[label distance=5]45:i}] (i) at (5.5,-1) {};
            \end{scope}
            \begin{scope}[every edge/.style={draw=black}]
            \path[very thick] (a) edge node {} (v);
            \path[very thick] (v) edge node {} (w);
            \path[very thick] (w) edge node {} (b);
            \path[dashed] (b) edge node {} (x);
            \path[very thick] (x) edge node {} (y);
            \path[very thick] (y) edge node {} (d);
            \path[dashed] (v) edge[bend left=35] node[label={[label distance=0]90:e}] {} (y);
            \path[very thick] (x) edge node {} (z);
            \path[very thick] (z) edge node {} (i);
            \path[very thick] (i) edge node {} (g);
            \path[dashed] (z) edge node {} (h);
            \end{scope}
            \begin{scope}[every node/.style={draw=none,rectangle}]
            \node (r) at (6.5,-2.5) {$R$};
            \end{scope}
            \end{tikzpicture}
        \caption{$R$ is a 3-ear, and $b$ and $g$ coincide}
        \label{case2bii-d}
        \end{subfigure}
        \caption{$x$ is adjacent to a vertex outside $P'$}
        \label{case2bii}
        \end{figure}
        
        \end{description}
    
    \end{description}
    \end{description}
    
    Since the above procedure takes constant time for every pair of pendant ears with adjacent internal vertices, the running time for the whole procedure is polynomial.
    
    On termination of this procedure, the ear-decomposition $D$ has $\phi(G)$ even ears, and is both open and nice.
\end{proof} 

\begin{lem} \label{otherlemma}
Let $D$ be an open nice ear-decomposition of a 2-vertex-connected graph $G$, and $M$ be the associated eardrum composed from the short (pendant) ears of D. Denote by $V_I$ the set of internal vertices of non-pendant ears, and let $\mu(G,M)$ be the size of the maximum earmuff for the eardrum $M$. Then $\mu(G,M) \leq |V_I| - 1$.
\end{lem}
\begin{proof}
    Suppose not. Then $\mu(G,M) \geq |V_I|$. Consider the graph $H$ on the vertex set $V_I$ with edge $(u,v)$ present in $E(H)$ if and only if there is a path with its endpoints at $u$ and $v$ in the maximum earmuff for $M$. Since $\mu(G,M) \geq |V_I|$, this graph has at least $|V_I|$ edges, and is hence not a forest. Since every edge in this graph corresponds to a path in the maximum earmuff, any cycle in this graph must be a cycle in the maximum earmuff, which contradicts the definition of an earmuff, which states that the union of all paths in the earmuff is a forest. Hence $\mu(G,M) \leq |V_I| - 1$.
\end{proof}

\begin{thm} \label{algorithmtheorem}
    There is a $\frac{17}{12}$-approximation algorithm for the minimum 2-vertex-connected spanning subgraph problem on graphs with minimum degree at least 3. For any 2-vertex-connected graph $G$ where every vertex has degree at least 3, it finds a 2-vertex-connected spanning subgraph with at most $\frac{17}{12}OPT_{2VC}(G)$ edges in polynomial time.
\end{thm}
\begin{proof}
    Construct an open evenmin nice ear-decomposition $D$ for $G$. Let $\pi$ denote the number of pendant ears and $\pi_3$ the number of (pendant) 3-ears in this ear-decomposition. We have $\pi_3 \leq \pi$. Let $H$ be the graph obtained by deleting from $G$ all edges that are in trivial ears in this ear-decomposition. Since the nontrivial ears of $D$ form an open ear-decomposition for $H$, $H$ is 2-vertex-connected (Whitney \cite{whitney}), and has at most $\frac{17}{12}$LP(G) edges, which we show using the following claims.
    
    \begin{inclm}
    The number of edges in nontrivial ears is at most $\frac{5}{4}LP(G) + \frac{1}{2}\pi$.
    \end{inclm}
    \begin{inproof}
    For any ear $P$ with $|E(P)| \geq 5$, we have $|E(P)| \leq \frac{5}{4}|in(P)|$. For any 4-ear or 2-ear $P$ we have $|E(P)| \leq \frac{5}{4}|in(P)| + \frac{3}{4}$. For any 3-ear $P$ we have $|E(P)| \leq \frac{5}{4}|in(P)| + \frac{1}{2}$.
    
    Let $E'$ be the set of edges in nontrivial ears. Since the total number of 4- and 2-ears in $D$ is at most $\phi(G)$, and $\pi_3 \leq \pi$, the total number of edges in nontrivial ears is at most $\frac{5}{4}(|V(G)| - 1) + \frac{3}{4}\phi(G) + \frac{1}{2}\pi \leq \frac{5}{4}L_\phi(G) + \frac{1}{2}\pi$, which is at most $\frac{5}{4}LP(G) + \frac{1}{2}\pi$.
    \end{inproof}
    
    \begin{inclm}
    The number of edges in nontrivial ears is at most $\frac{3}{2}LP(G) - \frac{1}{4}\pi$.
    \end{inclm}
    \begin{inproof}
    Since $D$ is an open nice ear-decomposition, the graph induced in $G$ by the internal vertices $V_M$ of the pendant short ears of $D$ has degree at most 1. Let $M$ be the set of its components, then $M$ is an eardrum in $G$. Let $V_D$ be the set of internal vertices of pendant long ears and let $V_I = V \setminus (V_M\cup V_D)$. Denote by $\phi_M$, $\phi_D$ and $\phi_I$ the number of even ears in the sets of pendant short ears, pendant long ears and non-pendant ears respectively.
    
    Let $E_1$ be the set of edges in pendant short ears. For every pendant short ear $P$, we have $|E(P)| = \frac{3}{2}|in(P)| + \frac{1}{2}\phi(P)$. Summing over all pendant short ears, we have $|E_1| = \frac{3}{2}|V_M| + \frac{1}{2}\phi_M$.
    
    Let $E_2$ be the set of edges in pendant long ears. For every pendant long ear $P$, we have $|E(P)| \leq \frac{3}{2}|in(P)| + \frac{1}{2}\phi(P) - 1$. Summing over all pendant long ears, we have $|E_2| \leq \frac{3}{2}|V_D| + \frac{1}{2}\phi_D - (\pi - |M|)$.
    
    Let $E_3$ be the set of edges in non-pendant ears. For every non-pendant ear $P$ except the single vertex ear $P_0$, since $P$ is a long ear, we have $|E(P)| \leq \frac{5}{4}|in(P)| + \frac{1}{2}\phi(P)$. For the vertex ear $P_0$, $|E(P_0)| = 0$ and $|in(P_0)| = 1$. Summing over all non-pendant ears including $P_0$, we have $|E_3| \leq \frac{5}{4}|V_I - 1| + \frac{1}{2}\phi_I$.

\allowdisplaybreaks[1]    
    Let $E' = E_1 \cup E_2 \cup E_3$ be the set of edges in nontrivial ears. Summing over the above inequalities, we get
    \begin{alignat*}{4}
    &|E'| && \quad \leq && \quad && \frac{3}{2}|V(G)| + \frac{1}{2}\phi(G) - \pi + |M| - \frac{1}{4}|V_I| - \frac{5}{4} \\
    &&& \quad = && \quad && \left[ |V(G)| + |M| - \mu(G,M) - 1 \right] \\
    &&& \quad && +\quad && \frac{1}{2}\left[ |V(G)| + \phi(G) - 1 \right] \\
    &&& \quad && -\quad && \pi \\
    &&& \quad && +\quad && \left( \frac{1}{4} + \mu(G,M) -\frac{1}{4}|V_I| \right) \\
    &&& \quad = && \quad && L_\mu(G,M) + \frac{1}{2} L_\phi(G) - \pi + \left( \frac{1}{4} + \mu(G,M) -\frac{1}{4}|V_I| \right) \\
    &&& \quad = && \quad && \frac{3}{2} LP(G) - \pi + \left( \frac{1}{4} + \mu(G,M) -\frac{1}{4}|V_I| \right) \\
    &&& \quad \leq && \quad && \frac{3}{2} LP(G) - \pi + \frac{3}{4}\mu(G,M) \qquad \text{using Lemma \ref{otherlemma}} \\
    &&& \quad \leq && \quad && \frac{3}{2} LP(G) - \pi + \frac{3}{4}\pi \qquad \text{since $\mu(G,M) \leq |M| \leq \pi$} \\
    &&& \quad \leq && \quad && \frac{3}{2} LP(G) - \frac{1}{4}\pi. \quad
    \end{alignat*}
    \end{inproof}
    (\textit{Proof of Theorem 3 continued})
    
    If $\pi \leq \frac{1}{3}LP(G)$, then from Claim 1, $|E'| \leq \frac{5}{4}LP(G) + \frac{1}{2}\pi \leq \frac{17}{12}LP(G) \leq \frac{17}{12}OPT_{2VC}(G)$.
    
    If $\pi > \frac{1}{3}LP(G)$, then from Claim 2, $|E'| \leq \frac{3}{2}LP(G) - \frac{1}{4}\pi < \frac{17}{12}LP(G) \leq \frac{17}{12}OPT_{2VC}(G)$.

    Applying Theorems \ref{pendanttheorem} and \ref{opennicetheorem} to $G$, and deleting all edges in trivial ears, we obtain a 2-vertex-connected spanning subgraph of cardinality at most $\frac{17}{12}OPT_{2VC}(G)$ in polynomial time.
\end{proof}

\clearpage
\section*{Appendix A}

\textbf{A1.} The 2-vertex-connectivity problem is NP-hard when restricted to graphs with minimum degree at least 3.
\begin{proof}
Let $G$ be an input graph to the general 2-vertex-connectivity problem, and denote by $n(G)$ the number of vertices with degree 2 in $G$.

Consider the graph $G'$ constructed as follows: replace every vertex with degree 2 in $G$ by an instance of $K_4$ (the complete graph on 4 vertices), such that the two edges incident on the degree-2 vertex in $G$ are incident on distinct vertices of the $K_4$ instance in $G'$. Then $G'$ has minimum degree at least 3, and every 2-vertex-connected spanning subgraph $H$ of $G$, with $|E(H)|$ edges, corresponds to a 2-vertex-connected spanning subgraph $H'$ of $G'$ (constructed by adding a path of length 3 between the degree-4 nodes of every $K_4$ instance created by replacement), with $|E(H')| = |E(H)| + 3n(G)$ edges, and vice-versa.

Hence any algorithm that solves the 2-vertex-connectivity problem in polynomial time on graphs with minimum degree at least 3 can be used to solve the unrestricted problem in polynomial time, which implies the statement of A1.
\end{proof}


\clearpage
\begin{thebibliography}{1}

\bibitem{khullervishkin} S. Khuller and U. Vishkin. ``Biconnectivity approximations and graph carvings''. In: \textit{Journal of the ACM (JACM)} 41.2 (1994), pp. 214-235.

\bibitem{gargsinglavempala} N. Garg, A. Singla, and S. Vempala. ``Improved  approximation  algorithms  for  biconnected  subgraphs  via  better lower  bounding  techniques''. In: \textit{Proc. 4th Annual ACM-SIAM SODA} (1993), pp. 103-111.

\bibitem{cheriyanseboszigeti} J. Cheriyan, A. Seb\H{o} and Z. Szigeti. ``Improving on the 1.5-Approximation of a Smallest 2-Edge Connected Spanning Subgraph''. In: \textit{SIAM J. Discrete Math.} 14.2 (2001), pp. 170-180.

\bibitem{sebovygen} A. Seb\H{o} and J. Vygen. ``Shorter tours by nicer ears: 7/5-approximation for the graph-TSP, 3/2 for the path version, and 4/3 for two-edge-connected subgraphs''. In: \textit{J. Combinatorica} 34.5 (2014), pp. 597-629.

\bibitem{whitney} H. Whitney. ``Non-separable and planar graphs''. In: \textit{Transactions of the American Mathematical Society} 34 (1932), pp. 339-362.

\end{thebibliography}
\end{document}